\documentclass[x11names,12pt]{article}
\usepackage[paper=letterpaper,margin=.85in]{geometry}

\pdfoutput=1

\usepackage{graphicx}				
\usepackage{amsmath}
\usepackage{mathtools}
\usepackage{amssymb}
\usepackage{booktabs}
\usepackage{comment}
\usepackage{amsthm}
\numberwithin{equation}{section}
\newtheorem{theorem}{Theorem}

\usepackage{multirow}
\usepackage{overpic}
\usepackage{caption}
\usepackage{subcaption}
\usepackage{array}
\usepackage{rotating}

\usepackage[utf8]{inputenc}
\usepackage{lmodern}
\usepackage[T1]{fontenc} 
\usepackage{microtype} 

\usepackage{xcolor}
\definecolor{dark-green}{rgb}{0.1,0.4,0}
\definecolor{NiceBlue}{rgb}{0.30196,0.55294,0.57647}

\usepackage{cite}
\usepackage{hyperref}
\hypersetup{
      colorlinks=true,
      linkcolor=dark-green,
      citecolor=Red4,
      urlcolor=dark-green,
      linktoc=page
}
\newcommand{\bea}{\begin{eqnarray}}
\newcommand{\eea}{\end{eqnarray}}

\newcommand\afterTocSpace{\bigskip\medskip}
\newcommand\afterTocRuleSpace{\bigskip\bigskip}

\newcommand{\dd}{\mathrm{d}}

\newcommand\emailfootnote[1]{%
  \begingroup
  \renewcommand\thefootnote{}\footnote{#1}%
  \addtocounter{footnote}{-1}%
  \endgroup
}

\theoremstyle{definition}

\begin{document} 
\begin{flushright}
  \footnotesize \color{dark-green}{UUITP-12/25}\\
  \normalsize
  \end{flushright}
\thispagestyle{empty}

\vspace*{1.5cm}
\begin{center}

{\bf {\LARGE Black Hole Singularities from \\ \vspace{5mm}Holographic Complexity}}

\begin{center}

\vspace{1cm}

{\bf Vyshnav Mohan}$^1$\emailfootnote{$^1$ \href{mailto:vyshnav.vijay.mohan@gmail.com}{vyshnav.vijay.mohan@gmail.com}}\\
 \bigskip \rm
  
\bigskip
\hspace{.05em}Science Institute,
University of Iceland \\Dunhaga 3, 107 Reykjav\'{i}k, Iceland.\\
\bigskip
\hspace{.05em} Institutionen för fysik och astronomi, Uppsala Universitet\\
Box 803, SE-751 08 Uppsala, Sweden
\bigskip

\rm
  \end{center}

\vspace{1.5cm}
{\bf Abstract}
\end{center}
\begin{quotation}
\noindent

Using a second law of complexity, we prove a black hole singularity theorem. By introducing the notion of trapped extremal surfaces, we show that their existence implies null geodesic incompleteness inside globally hyperbolic black holes. We also demonstrate that the vanishing of the growth rate of the volume of extremal surfaces provides a sharp diagnostic of the black hole singularity. In static, uncharged, spherically symmetric spacetimes, this corresponds to the growth rate of spacelike extremal surfaces going to zero at the singularity. In charged or rotating spacetimes, such as the Reissner-Nordström and Kerr black holes, we identify novel timelike extremal surfaces that exhibit the same behavior at the timelike singularity.
\end{quotation}

\setcounter{page}{0}
\setcounter{tocdepth}{2}
\setcounter{footnote}{0}
\newpage

\parskip 0.1in
 
\setcounter{page}{2}

\setcounter{tocdepth}{1}
\tableofcontents
\afterTocSpace
\hrule
\afterTocRuleSpace

\section{Introduction}
\label{introsec}
Penrose proved that a globally hyperbolic, inextendible spacetime manifold $M$ is null geodesically incomplete if it contains a trapped surface \cite{Penrose:1964wq}. A surface is future-trapped if the future-directed null geodesics orthogonal to the surface converge. The proof crucially assumes the null energy condition, but there exist states in relativistic quantum field theories that violate this condition \cite{Epstein:1965zza}. Wall \cite{Wall:2010jtc} showed that the energy condition can be given up in favour of the generalized second law, which states that the generalized entropy always increases. The generalized entropy of a codimension-two surface is defined as the following sum
\bea
S_\text{gen} = \frac{A}{4G_N \hbar} +S_{\text{out}}\,, \label{genentropyeq}
\eea
where $A$ is the area of the surface and $S_{\text{out}}$ is the entanglement entropy of the quantum fields in the exterior of the surface. The generalized second law reflects the thermal properties of the black hole. Since no counterexamples to the generalized second law are known, Wall's proof puts the singularity theorem on a firmer footing. Significant advancements in the last decade have further refined the singularity theorem, resulting in wider applicability (see \cite{Bousso:2025xyc} and references therein).

The study of quantum complexity has also led us to another version of the second law of black hole thermodynamics. The \emph{Second Law of Complexity} (SLC) states that if the complexity of a black hole is below its saturation value---which is of the exponential order in the Bekenstein-Hawking entropy---it will always increase \cite{Brown:2016wib,Brown:2017jil}. The complexity of black holes can be holographically realized through the ``Complexity=Volume'' (CV) prescription, where complexity is defined as the volume of boundary-anchored extremal co-dimension one surface passing through the interior of the black hole \cite{Susskind:2014rva,Stanford:2014jda}. Therefore, it is natural to ask whether we can use the CV prescription in conjunction with the second law of complexity to arrive at a singularity theorem. In this paper, we show that the answer is affirmative.

We begin by proposing a generalization of the second law of complexity: the volume of extremal surfaces anchored to a causal horizon increases as the anchor points are pushed to the future (see Section \ref{theoremsec} for more details). While we do not provide a proof, we take this statement as an assumption, supported by various physically relevant examples. The second law of complexity provides a powerful \emph{bulk} condition that can be used to establish a singularity theorem. The following heuristic picture explains the interplay between singularities and SLC. Consider a sheet of future-directed null rays in the interior of a black hole. Let us look at extremal volume surfaces confined to the region between these null rays. The volume of these surfaces decreases as we move it toward the singularity, as shown below:
\begin{figure}[h!]
  \centering
    \hspace{4.2cm}\includegraphics[width=0.5\linewidth]{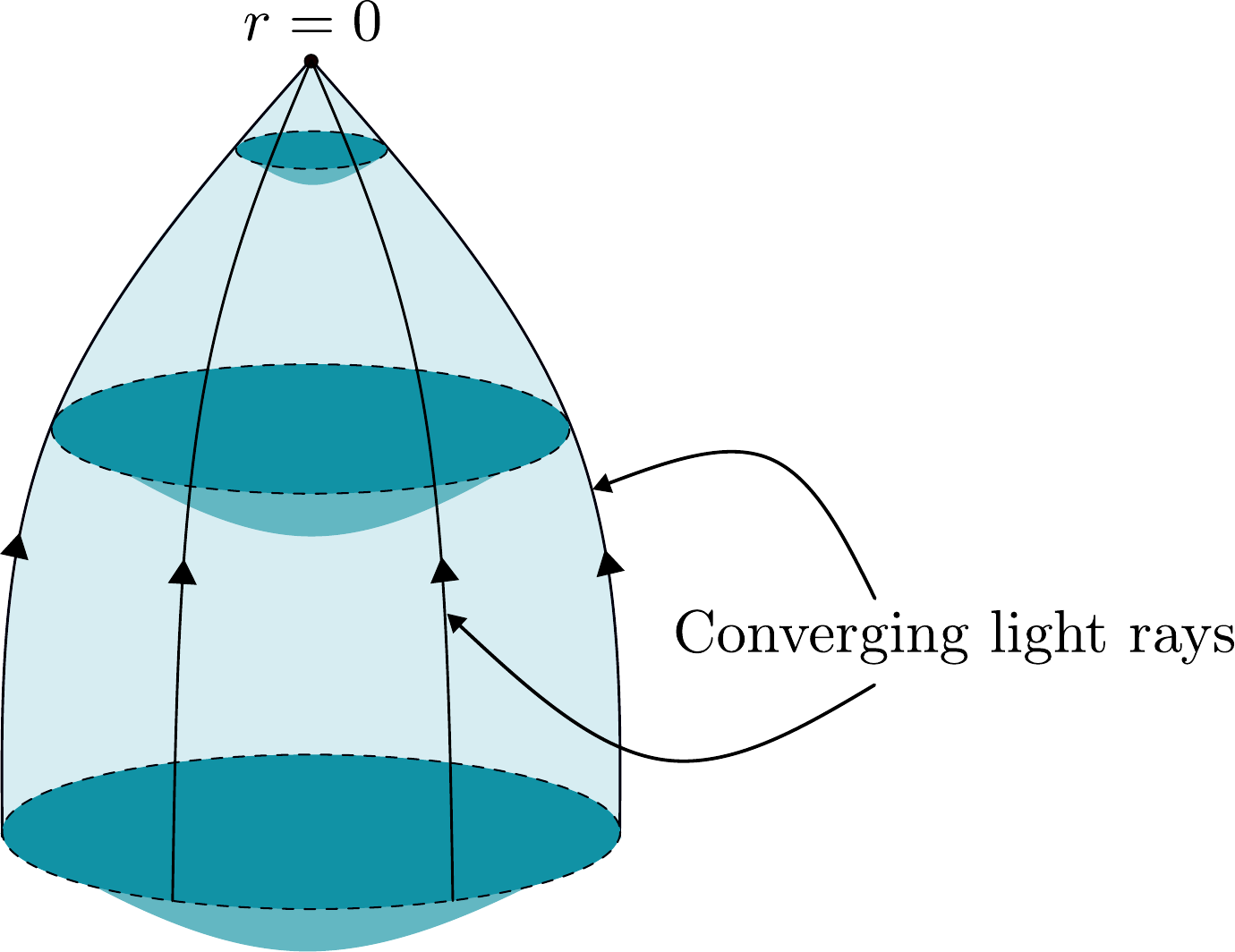}
  \caption{}
  \label{nullfocusfig}
  \end{figure}

We refer to these extremal surfaces as trapped extremal surfaces. If the null rays were infinitely extendable, they would form a causal horizon, and the SLC would then force the volume of extremal surfaces to increase. However, focusing the light rays prevents this, leading to geodesic incompleteness in the black hole interior.

An interesting consequence of this picture is the observation that there is no `space' for the extremal surfaces to grow at the singularity. As a result, the vanishing growth rate of the volume serves as a robust probe for diagnosing the existence of the singularity. In the case of a static uncharged spherically symmetric black hole, this translates to the statement that the volume of spacelike extremal surfaces vanishes at the singularity. In Section \ref{Cauchysec}, we show that the growth rate of timelike extremal surfaces, not their spacelike counterparts, goes to zero at the singularity in Reissner-Nordström and Kerr black holes.

The remainder of the paper is organized as follows. In Section \ref{staticbhsec}, we study the volume growth of extremal surfaces anchored to null surfaces. We introduce the notion of trapped extremal surfaces and prove a singularity theorem in Section \ref{theoremsec}. In Section \ref{Cauchysec}, we examine the growth of extremal surfaces in Reissner-Nordström and Kerr black holes, identifying novel timelike extremal surfaces whose growth rate vanishes at the singularities. In Section \ref{rindlersec}, we make the important observation that SLC implies null geodesic incompleteness, which does not always correspond to the presence of a singularity. We explore this distinction by studying complexity growth in AdS-Rindler, where incompleteness arises from the coordinate chart's failure to cover the entire Poincaré patch.
\section{Growth of Extremal Surfaces}
\label{staticbhsec}
\begin{figure}
  \centering
    \includegraphics[width=0.75\linewidth]{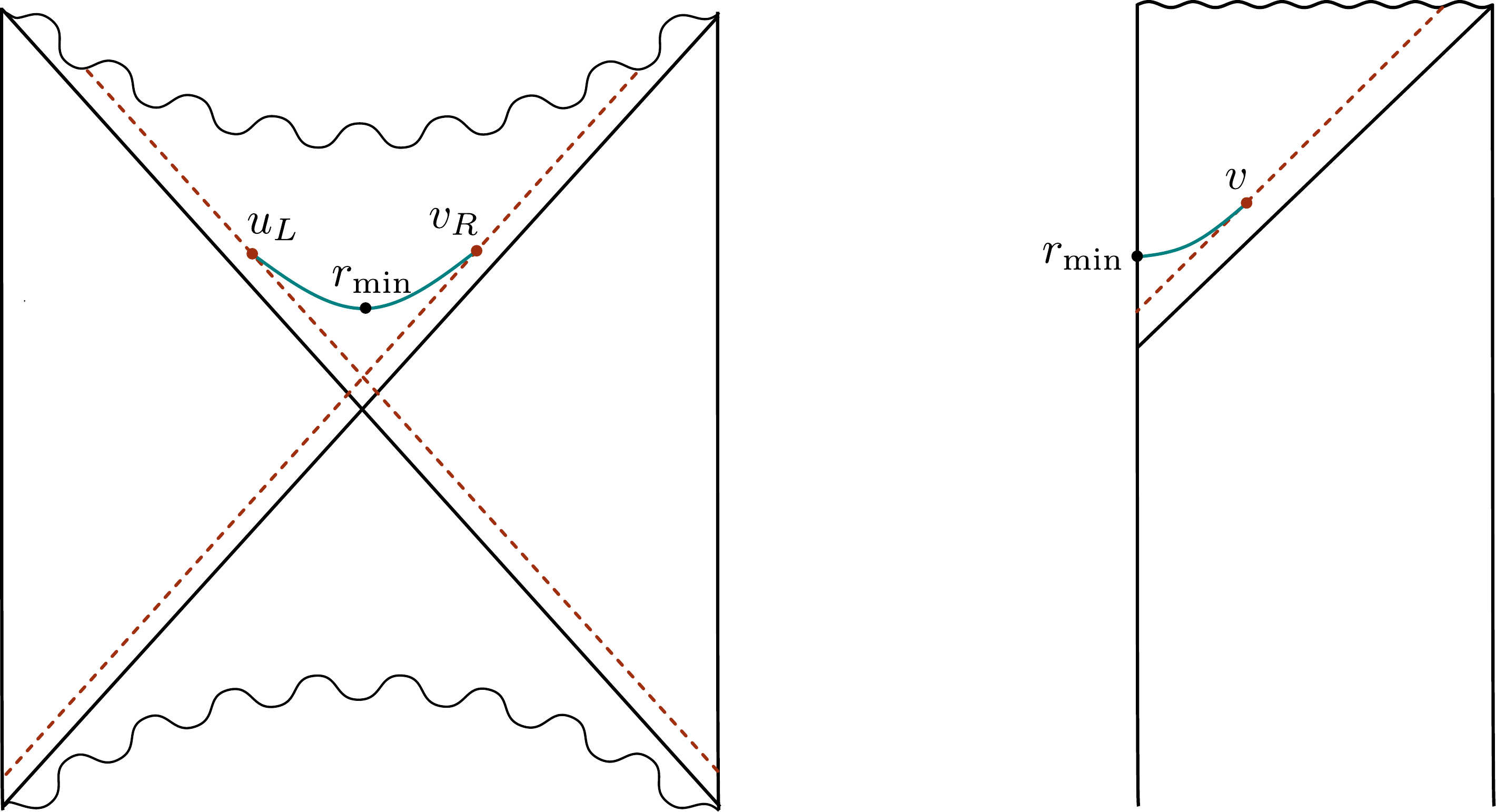}
  \caption{\small Extremal volume surfaces anchored to null surfaces are shown, with anchor surfaces indicated by red dashed lines. We have displayed the extremal curves in the Penrose diagrams of an eternal (left) and a single-sided black hole (right). In both cases, the location of the null surfaces and the anchoring point can be conveniently specified by using Eddington-Finkelstein coordinates $u,v$. The turning point of the surface is located at $r=0$.}
  \label{gencompexamplefig}
  \end{figure}
In this section, we calculate the growth rate of extremal surfaces anchored to null surfaces inside a black hole (see Figure \ref{gencompexamplefig}). In the standard treatment, the extremal surfaces are anchored to a constant radial surface in the asymptotic region (see \cite{Carmi:2017jqz} for a review). Our discussion here differs from the conventional approach in our choice of boundary conditions. In an eternal black hole, we anchor the extremal surface onto two symmetric null surfaces in the interior of the black hole. In the case of a single-sided black hole, we instead use the bridge-to-nowhere prescription \cite{Susskind:2014jwa} where the extremal surface ends at $r=0$. Since the anchor surface falls into the black hole, the growth rate of the extremal surfaces differs from the standard behavior, as discussed below.

Consider an uncharged $(d+1)$-dimensional static eternal black hole with the metric
\bea
\dd{s}^2=-f(r) \dd{t}^2+\frac{\dd r^2}{f(r)}+r^2 \dd{}\Omega_{d-1}^2\,,\label{staticbhmetric}
\eea 
where $f(r)$ is the blackening factor. The black hole spacetime is allowed to asymptote to AdS, dS, or flat space. The event horizon is located at $r=r_h$. The Eddington-Finkelstein coordinates $u$ and $v$ are defined as follows
\bea
v = t+r^{*},\quad u = t-r^{*}, \qquad \text{where} \ \frac{\dd r^{*}}{\dd r} = \frac{1}{f(r)}\,.
\eea 
We then specify the location of the null surfaces by choosing their corresponding Eddington-Finkelstein coordinates to be a constant, which we will denote by $u_0$.

The volume of a spherically symmetric codimension-one surface anchored onto these null surfaces is given by
\bea
\mathcal{V} = \Omega_{d-1}\int \dd\lambda \ r^{d-1}\sqrt{-f(r) \dot{v}^2+2\dot{v}\dot{r}}\,.\label{ogvolintegral}
\eea  
Here, we have used $\lambda$ to parametrizes the surface as $(v(\lambda), r(\lambda))$. We have also denoted the volume of a $(d-1)$-dimensional unit sphere by $\Omega_{d-1}$. We anchor the surface at symmetric points by choosing the coordinates of their endpoints to be $u_L=v_R=v$ (see Figure \ref{gencompexamplefig}). Furthermore, we will use $r_v$ to denote the value of the $r$ coordinate at the anchor points.

The extremal volume surfaces are then obtained by extremizing the volume functional. Since the volume integral is reparametrization invariant, we choose $\lambda$ by imposing the following condition,
\bea
r^{d-1} \sqrt{-f \dot{v}^2+2 \dot{v} \dot{r}}=1\,,\label{reparametrizationeq}
\eea
and this gives us the equations of motion
\bea
\begin{aligned}
p &=r^{2(d-1)}(f(r) \dot{v}-\dot{r})\,, \label{eometernal}\\
r^{2(d-1)} \dot{r}^2 &=f(r)+r^{-2(d-1)} p^2\,.\label{eometernal1}
\end{aligned}
\eea
Here, $p$ is a conserved quantity along the extremal surface. The location of the turning point $r_{\min}$ can be obtained by setting $\dot{r}$ to zero in the above expression
\bea
r^{2d-2}_{\min}f(r_{\min}) +p^2=0\,. \label{turningpointeq}
\eea 
We choose $\lambda$ to increase from left to right in the Penrose diagram. Evaluating the equations of motion at the turning point, we find that $p$ is negative.

Performing a change of coordinates, we can recast the second equation in \eqref{eometernal1} as a particle scattering off a Newtonian potential $V_{\text{eff}} = - r^{2d-2}f(r)$ \cite{Carmi:2017jqz,Gautason:2025ryg}:
\bea
\dot{\tilde{r}}^2 + V_{\text{eff}} = p^2, \quad \ \text{where} \quad \ \frac{\dd \tilde{r}}{\dd r} = r^{2d-2}\,.
\eea 
Here, $p^2$ plays the role of the total energy of the particle. The scattering picture is handy in characterizing the extremal surfaces since we do not have explicit closed-form expressions. The effective potential has a maximum at some $r=r_a$, which we refer to as the \emph{accumulation surface}. The region around this surface plays a crucial role in the usual holographic volume complexity computations \cite{Stanford:2014jda,Susskind:2014rva,Carmi:2017jqz,Gautason:2025ryg}. 
\begin{figure}
  \centering
  \hspace*{0.2cm} \includegraphics[width=1\linewidth]{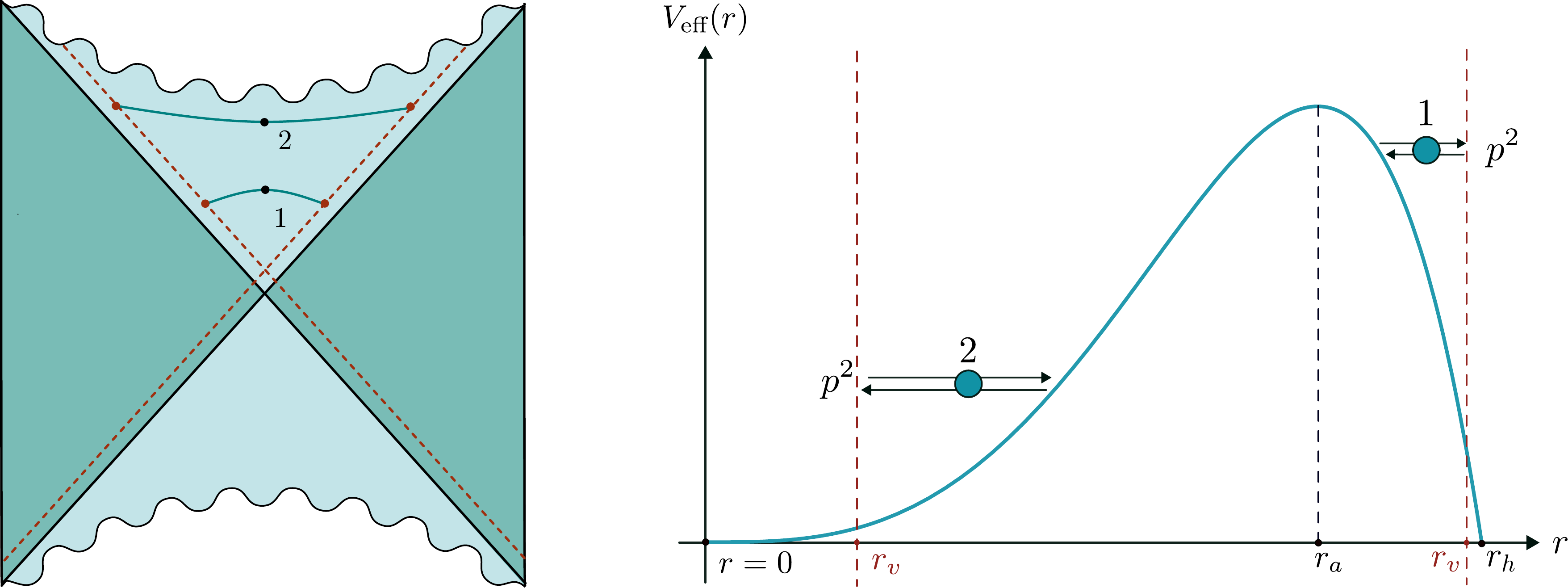}
  \caption{\small \emph{(Left)} Two extremal surfaces in the interior of a static spherically symmetric uncharged black hole. \emph{(Right)} The corresponding trajectories in the scattering picture. The red dashed lines are the locations of the anchor surface.}
  \label{EternalBHfig}
  \end{figure}

Now, let us read off the properties of the extremal curves using the effective potential. Consider the case in which the anchor point is near the bifurcation surface of the black hole, i.e., $r_v$ is close to $r_h$. The associated extremal surface is represented by curve $1$ in the left panel of Figure \ref{EternalBHfig}. In the scattering picture, the surface corresponds to a particle scattering off the right side of the potential (see curve 1 in the right panel of Figure \ref{EternalBHfig}). The particle starts at $r=r_v$, reaches a turning point located at $r_a<r_{\min}<r_v$, and then returns to the anchor surface.

As we move the anchor point along the null surface, $v$ increases, and $r_v$ decreases. The energy of the particle $p^2$ increases and the turning point approaches the top of the potential. When $r_v=r_a$, the extremal surface becomes a constant-$r$ surface. In the scattering picture, this corresponds to a particle sitting at the top of the effective potential.

For any larger value of $v$, $p^2$ decreases. The extremal surfaces beyond this point correspond to a particle scattering off the left side of the effective potential (see curve $2$ in Figure \ref{EternalBHfig}). This ensures that $\dot{r}$ in equation \eqref{eometernal1} takes only real values.

Now, let us compute the volume of the corresponding surfaces. Since $\lambda$ increases from left to right in the Penrose diagram, we can use \eqref{reparametrizationeq} to rewrite the volume integral \eqref{ogvolintegral} as follows:
\bea
\mathcal{V}=2\Omega_{d-1}\int_{\lambda(r=r_{\min})}^{\lambda(r=r_v)} \dd \lambda = 2 \Omega_{d-1} \int_{r_{\min}}^{r_v} \frac{\dd r}{\dot{r}}\,.
\eea 
We have used the symmetry about the $t=0$ surface to simplify the expression. Using the equations of motion, we get
\bea
\mathcal{V}=\pm 2 \Omega_{d-1} \int_{r_{\min}}^{r_v} \dd r \frac{r^{2(d-1)}}{\sqrt{f(r) r^{2(d-1)}+p^2}}\,. \label{volumeexpressioneternalfinal}
\eea 
The choice of sign is determined by the sign of $\dot{r}$. When $r_{\min} \geq r_a$, $\dot{r}$ is positive while it is negative for the extremal surfaces with $r_{\min} <r_a$. This reverses the limits of integration. In the scattering picture, this corresponds to choosing the opposite orientation of the particle.

Using \eqref{eometernal}, we can express $v$ as another $r$ integral: 
\bea
v-r^{*}_{\min} = \int_{r_{\min}}^{r_v} \dd r\left[\frac{p}{f(r) \sqrt{f(r) r^{2(d-1)}+p^2}}+\frac{1}{f(r)}\right]\,. \label{vofp}
\eea 
Here $r^{*}_{\min}$ is the location of the turning point in the tortoise coordinate. Since $r_{\min}$ is related to $p$ through \eqref{turningpointeq}, $v$ depends only on $p$. Therefore, the anchoring point $v$ fixes a value of $p$, which in turn fixes the location of the turning point.

Now, let us calculate the growth rate of these surfaces. When $r_{\min} \geq r_a$, we can use \eqref{volumeexpressioneternalfinal} and \eqref{vofp} to arrive at the following expression \cite{Carmi:2017jqz}
\bea
\frac{\mathcal{V}}{2 \Omega_{d-1}}=\int_{r_{\min }}^{r_{v}} \dd r\left[\frac{\sqrt{f(r) r^{2(d-1)}+p^2}}{f(r)}+\frac{p}{f(r)}\right]-p\left(v-r^*_{\min} \right)\,.
\eea 
Taking a derivative of the above expression w.r.t $v$, we find that
\bea 
\frac{1}{2\Omega_{d-1}} \frac{\dd\mathcal{V}}{\dd v} = \frac{\dd r_v}{\dd v}\left[\frac{\sqrt{f(r_v) r_v^{2(d-1)}+p^2}}{f(r_v)}+\frac{p}{f(r_v)}\right]-p\,,
\eea 
We have used \eqref{vofp} to simplify the expression. Using $r^{*}_v=(v-u_0)/2$, we get
\bea 
\frac{\dd\mathcal{V}}{\dd v} = \Omega_{d-1} \left(-p+\sqrt{f(r_v) r_v^{2(d-1)}+p^2}\right)\,.\label{volumegrowtheternal}
\eea 
Since $p$ is negative, the extremal surfaces grow with $v$. This growth, however, ends when the turning point reaches the accumulation surface. When $r_{\min}<r_a$, we can follow the same steps, and find that
\bea 
\frac{\dd \mathcal{V}}{\dd v} = \Omega_{d-1} \left(p-\sqrt{f(r_v) r_v^{2(d-1)}+p^2}\right)\,.\label{volumegrowtheternal2}
\eea
The growth rate is now negative, and the volume decreases when $v$ increases.

When the anchor point reaches the singularity, the growth rate becomes zero. This is because the corresponding extremal surface is just the $r=0$ surface and $p$ vanishes on it. Since the growth rate vanishes only at the $r=0$ surface, we can use it as an effective diagnostic of the black hole singularity. An important consequence of this observation is that once the anchor point passes the accumulation surface, extremal surfaces are inevitably `doomed' to fall into the singularity where their volume growth goes to zero. The accumulation surface, therefore, acts as a \emph{trapping} surface. We will use this crucial fact to prove a singularity theorem in the next section.
\section{Geodesic Incompleteness from Second Law}
\label{theoremsec}
The location of the anchor surfaces in the previous section can be generalized using the causal structure of spacetime. Consider some point $p$ on the spacetime manifold $M$. We define its chronological future $I^{+}(p)$ as the set of all points in $M$ that can be reached by future-directed timelike curves passing through $p$. Then, by definition, the boundary of $I^{+}(p)$, which we will denote by $\partial I^{+}(p)$, is an achronal boundary \cite{Hawking:1973uf}. Moreover, $N \equiv \partial I^{+}(p)$ is generated by null geodesics emanating from $p$ (see Theorem 8.1.2 in \cite{Wald:1984rg}). The extremal surfaces we studied in the previous section can be understood as being anchored to $N$ for some choice of $p$. In the case of a static black hole, we chose $p$ to be a point lying on the $t=0$ surface.

Now let us define a notion of a \emph{causal horizon}. Consider a future-infinite timelike worldline $W$. We define a future-causal horizon $H_\text{fut}$ as the boundary of the past of the worldline, that is, $H_{\text{fut}} = \partial I^{-}\left(W\right)$ \cite{Jacobson:2003wv}. It is important to emphasize that $W$ is future-infinite; that is, the proper time along the worldline runs all the way to infinity. This distinguishes $H_\text{fut}$ from the future boundary of the domain of dependence of arbitrary subregions. We can then quickly verify that this definition of causal horizons includes black hole event horizons as well as observer-dependent horizons, such as Rindler and cosmological horizons.

Using the definition of a causal horizon, we rephrase the second law in \cite{Brown:2016wib,Brown:2017jil} as follows:

\textbf{Second Law of Complexity} (SLC): The volume of extremal surfaces anchored to a causal horizon $H_{\text{fut}}$ always increases as we move the anchor point along the surface into the future.

It is important to note that we do not have a proof of SLC. Explicit computations of holographic complexity in various cases, for example, in static black holes \cite{Carmi:2017jqz}, Rindler wedges, rotating black holes \cite{Couch:2018phr}, even non-static black holes \cite{Chapman:2018dem,Chapman:2018lsv} all point towards the validity of SLC\footnote{Readers familiar with holographic complexity computations might notice that the SLC used here differs from standard definition, as the extremal surface is anchored to the horizon rather than to a cutoff surface in the asymptotic region. However, we can quickly convince ourselves that the growth of these surfaces comes from the turning point, which lies well inside the horizon, and that both formulations are classically equivalent \cite{Carmi:2017jqz} (see also \cite{Stanford:2014jda,Schneiderbauer:2019anh,Gautason:2025ryg} for related discussions).}. We will not attempt to prove SLC here. Instead, we assume its validity.

The second law, as stated here, is coarse-grained in that we have ignored quantum corrections to it. It is known that there are corrections non-perturbative in the exponential of the Bekenstein-Hawking entropy, which will cause the saturation of the volume \citeleft\citen{Iliesiu:2021ari}\citepunct\citen{Balasubramanian:2022gmo}\citepunct\citen{Gautason:2025ryg}\citeright. We ignore these corrections for now and revisit them in Section \ref{Discussionsection}.

Now, let us define a future \emph{trapped extremal surface} as an extremal surface anchored to $N$ with the property that its volume decreases as the anchor point is moved into the future. We see that the extremal surfaces in the previous section with $r_{min}<r_a$ are trapped extremal surfaces. An important point to note here is that the existence of these surfaces does not violate SLC, as the anchoring surfaces are not causal horizons. Also, by definition, trapped extremal surfaces are compact.

In this section, we restrict ourselves to globally hyperbolic spacetimes, where (a) there are no closed timelike curves, and (b) for any two points $p$ and $q$, $J^{+}(p) \cap J^{-}(q)$ is compact, where $J^{\pm}(p)$ is the causal future/past of $p$.

The statement of the SLC turns out to be extremely powerful because it provides a \textit{bulk} condition on the growth rate of volumes anchored to causal horizons. An immediate implication of the SLC is that if there are trapped extremal surfaces somewhere in the spacetime, then they \emph{cannot} be anchored to a causal horizon. Now, let us use this fact to arrive at a singularity theorem.
\begin{theorem}
Consider a globally hyperbolic spacetime $M$ with a trapped extremal surface $T$. Then, the spacetime is null geodesically incomplete.
    \end{theorem} 
\begin{proof} 
Let us assume the trapped extremal surface $T$ is anchored to a null surface $N=\partial I^{+}(p)$, for some point $p \in M$. We will suppose for a contradiction that the affine parameter of every null geodesic generating $N$ can be extended infinitely into the future. This would then imply that $N$ is a causal horizon. From SLC, the volume of the trapped surface would always increase, contradicting the defining property of a trapped extremal surface. Therefore, the affine parameter of the null geodesics on $N$ must terminate at a finite value.

Now, we can complete the proof using the original Penrose theorem \cite{Penrose:1964wq} (see \cite{Hawking:1973uf,Wall:2010jtc} for more details). Suppose $M$ were null geodesically complete. Then the null geodesic generators of $N$ could be extended to the future beyond $N$. An important point to note is that the affine parameter \emph{on} $N=\partial I^{+}(p)$ is bounded, and any extension of the null generators to the future would lie inside $I^{+}(p)$, not on its boundary (see Proposition 4.5.14 and Section 8.2 of \cite{Hawking:1973uf} for a related discussion). Null geodesic completeness of $M$ implies that the endpoints of the null geodesic generators of $N$ lie within the spacetime and therefore remain part of $N$, ensuring that $N$ is closed.

We now argue that $N$ is compact. Since the affine parameters are bounded, we can rescale them to lie within the range $[0,1]$. Then $N$ can be topologically written as $T \times [0,1]$, which is a product of a compact space and a closed interval, making $N$ compact.

This conclusion contradicts the assumption that the spacetime is globally hyperbolic. To see this, consider a connected, non-compact Cauchy slice $\Sigma$. We can always find a smooth timelike vector field $t^{\mu}$ whose integral curves intersect $\Sigma$ exactly once. Because $N \equiv \partial I^{+}(p)$ is achronal, these integral curves intersect $N$ at most once as well. This allows us to define a one-to-one map between $\Sigma$ and $N$ using the integral curves of $t^{\mu}$. However, this leads to a contradiction since the global hyperbolicity of the spacetime prohibits the evolution of a non-compact slice $\Sigma$ into a compact surface $N$. Therefore, the assumption that $M$ is null geodesically complete, which was used to argue that $N$ is compact, is incompatible with the other assumptions.
\end{proof}
Before we conclude the discussion, a few comments are in order:
\begin{itemize}
  \item Our theorem follows the same structure as other singularity theorems \cite{Senovilla:1998oua}. Specifically, we impose a `causality condition' by assuming global hyperbolicity of the spacetime, which helps rule out important edge cases (see, e.g., \cite{Borde:1994ai}). The second law of complexity replaces the usual `energy condition,' while trapped extremal surfaces provide the `initial condition.'
  \item In the standard CV prescription, complexity is assigned to the volume of a surface that lies outside the accumulation surface. No such prescription exists for trapped extremal surfaces located inside the accumulation surface. While it may be possible to define a notion of complexity for these trapped surfaces, the proof does not rely on such a definition. The SLC is carefully phrased to depend only on the volume of extremal surfaces anchored to causal horizons. As a result, we avoid the need to interpret trapped extremal surfaces in terms of complexity and do not rely on any non-standard notions of holographic complexity.
\end{itemize}

\section{To Cauchy Horizons and Beyond}
\label{Cauchysec}
In Section \ref{staticbhsec}, we studied the growth of extremal surfaces in static, uncharged black holes. In this section, we will extend this discussion to the case of a charged black hole. Since the metric of these black holes can also be brought to the static form in \eqref{staticbhmetric}, similar computations are expected to hold. However, the presence of a Cauchy horizon in the interior brings in new surprises.

\subsection{Reissner-Nordstr\"{o}m Black Hole}
\label{RNsec}
Consider a Reissner-Nordstr\"{o}m (RN) black hole in asymptotically flat spacetime. We will work in $3+1$-dimensions. The black hole has the metric \eqref{staticbhmetric}, with the following blackening factor
\bea
f(r) = \left(1-\frac{r_+}{r}\right)\left(1-\frac{r_-}{r}\right)\,,
\eea
where $r_{\pm}$ are the location of the horizons with $r_+>r_-$. The inner horizon is a Cauchy horizon and signals the breakdown of global hyperbolicity \cite{Hawking:1973uf}. The original Penrose theorem and its entropic reformulations do not apply to such black holes because of this breakdown of global hyperbolicity.
\begin{figure}
  \centering
  \hspace*{-0.9cm} \includegraphics[width=1.1\linewidth]{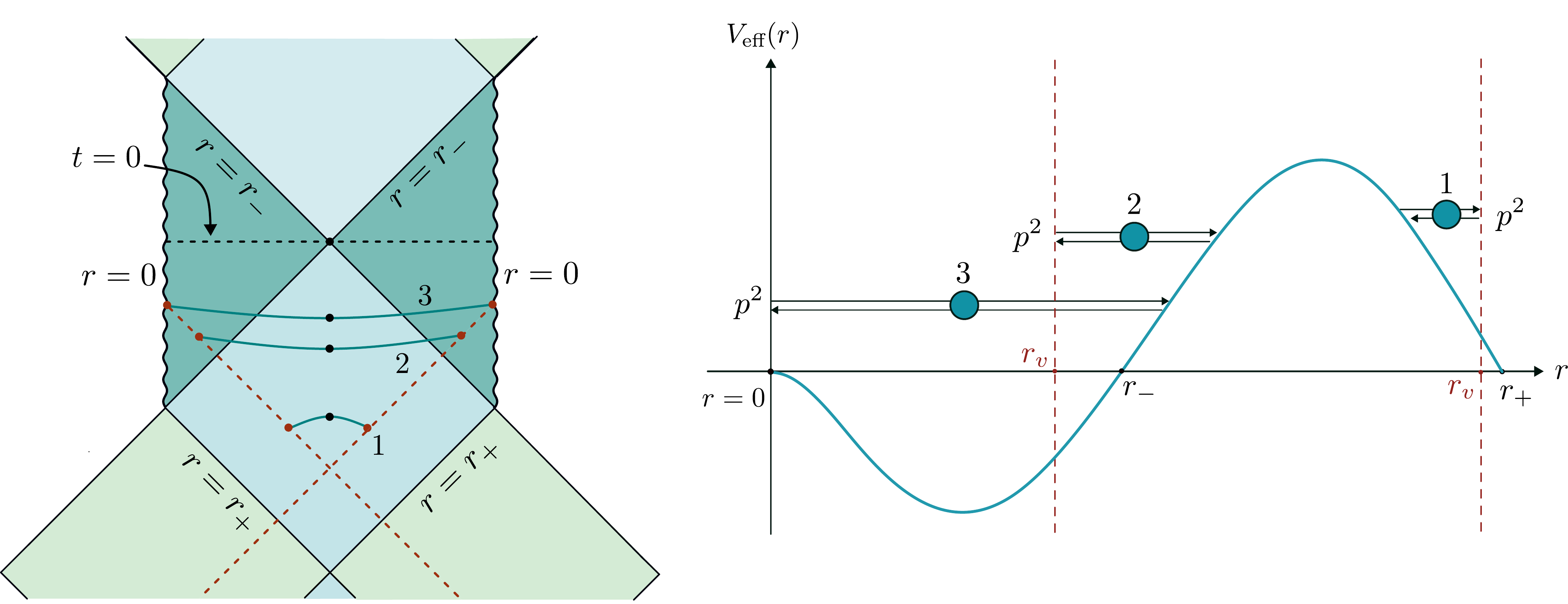}
  \caption{\small \emph{(Left)} Spacelike extremal surfaces in a Reissner-Nordström black hole. \emph{(Right)} The corresponding trajectories in the scattering picture. The anchor surface of curve 3 is located at the singularity.}
  \label{RNspacelikefig}
  \end{figure}

Now let us study the growth rate of \emph{spacelike} extremal volume surfaces. As in the previous section, we will look at spherically symmetric extremal volume surfaces anchored onto two symmetric null surfaces in the interior of the black hole. The equations of motion for these surfaces are once again given by \eqref{eometernal}. The extremal volumes in the interior of the Reissner-Nordström (RN) black hole were studied in \cite{PhysRevD.31.1267}, and we have illustrated a few of these surfaces in Figure \ref{RNspacelikefig}.

In the scattering picture, these surfaces can be thought of as particles scattering off the effective potential $V_{\text{eff}}$ (see right panel of Figure \ref{RNspacelikefig}). As we move the anchor point towards the singularity,  the volume of extremal surfaces initially increases but begins to decrease once the anchor point crosses the accumulation surface. The volume growth rates are given by \eqref{volumegrowtheternal} and \eqref{volumegrowtheternal2}. The crucial difference from the uncharged case is that $p$ never goes to zero, even when the anchoring point has reached the singularity.

To see this, consider \eqref{volumegrowtheternal} where the growth rate can be seen to be proportional to $p$. From the first equation in \eqref{eometernal}, we observe that $p$ vanishes if and only if $f(r)=0$ or $t=$ constant along the extremal surface. Since the condition $f(r)=0$ is never satisfied, the only possibility is to consider null anchoring surfaces that reach the singularity at the $t=0$ slice (depicted by the horizontal dotted line in the Penrose diagram in Figure~\ref{RNspacelikefig}). For any other general anchoring surface, the growth rate remains nonzero. Therefore, the fate of the trapped extremal surfaces is not as doomed as in the previous case, as their volumes never shrink to zero when the anchor point reaches the singularity. This is because the null anchoring surfaces go off to timelike $r=0$ slices in two \emph{different} coordinate patches. This suggests that the vanishing of the volume of extremal surfaces ceases to be a suitable probe of the singularity.

However, we can consider an alternate computation that fixes this problem. We do not anchor both the endpoints of the extremal surface to the null surfaces. Instead, we will anchor only the `right' endpoint and let the `left' one terminate at its turning point on the $t=0$ line. Now, let us repeat the extremal surface computation. When $r_-<r_v<r_a$, we obtain precisely half of the spacelike extremal surfaces (see curve 1 in Figure \ref{RNBHfig}). As discussed earlier, the growth rate of the extremal surfaces is always positive.

When the anchor point crosses the inner horizon, we find new behaviour. There are now two possible choices for the extremal surfaces. We have the usual spacelike surface (curve 2 in Figure \ref{RNBHfig}). Additionally, we get a timelike surface that satisfies the same boundary conditions (curve 3 in Figure \ref{RNBHfig}). The orientation of the scattering trajectory in Figure \ref{RNBHfig} is chosen to ensure that $\dot{r}$ remains real. The lack of a unique extremal surface is a consequence of the breakdown of global hyperbolicity of the spacetime. The growth rate of these timelike surfaces can be computed using the techniques in Section \ref{staticbhsec}. This time, we choose the parameter along the surface as follows:
\bea
r^{d-1} \sqrt{-f \dot{v}^2+2 \dot{v} \dot{r}}=-1\,,\label{reparametrizationtimeeq}
\eea
where the negative sign on the right-hand side comes from the fact that we are working with timelike surfaces. Denoting the conserved quantity along the surface by $p$, we arrive at the following equations of motion
\bea
\begin{aligned}
p &=r^{2(d-1)}(f(r) \dot{v}-\dot{r})\,, \\
r^{2(d-1)} \dot{r}^2 &=-f(r)+r^{-2(d-1)} p^2\,.\label{eometernaltime}
\end{aligned}
\eea
Note that the equations of motion can also be obtained from \eqref{eometernal} by the simple substitution $p \to i p$ and $\lambda \to i \lambda$. Choosing $\lambda$ to increase in the downward direction in the Penrose diagram, we find that $p\geq 0$.
\begin{figure}
  \centering
  \hspace*{-0.9cm} \includegraphics[width=1.1\linewidth]{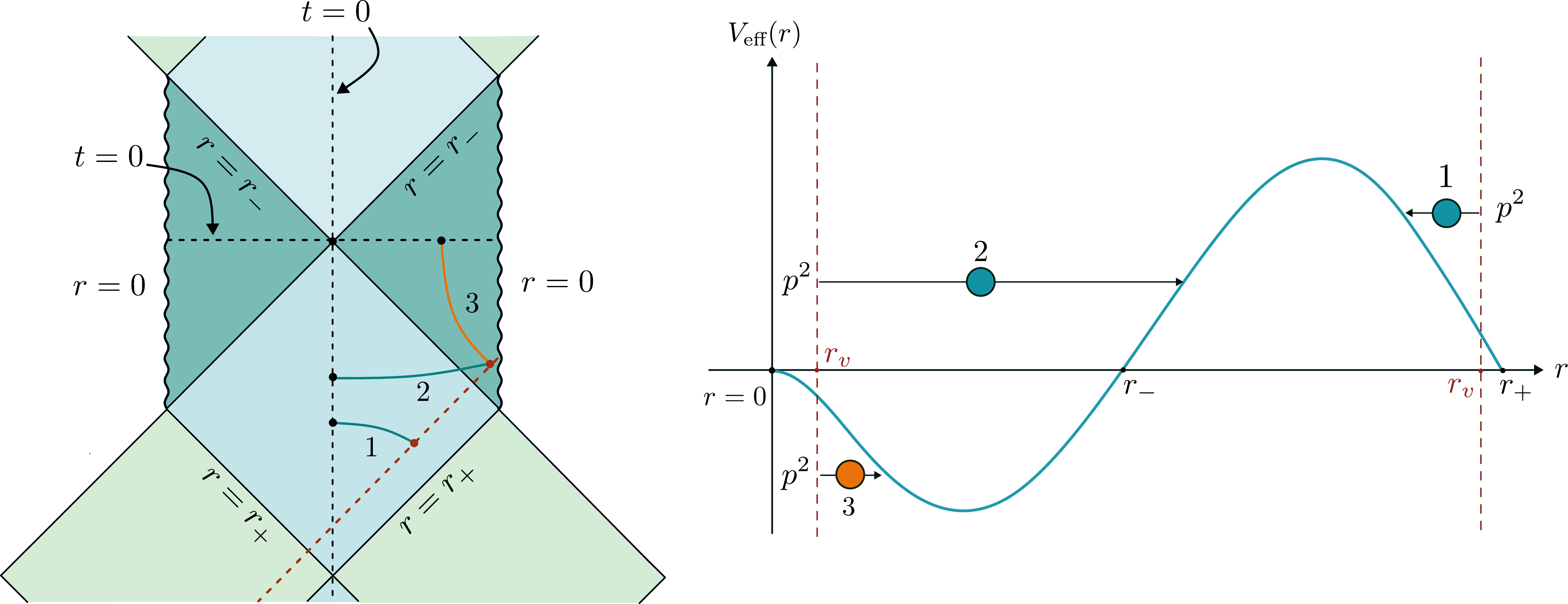}
  \caption{\small The teal curves correspond to spacelike extremal surfaces, while the orange curve is timelike. The red dashed lines once again correspond to the locations of the anchor surface.}
  \label{RNBHfig}
\end{figure}

We can once again rewrite the equation involving the $r$ coordinate as a particle scattering of an effective potential $\tilde{V}_{\text{eff}}$. The new effective potential turns out to be the negative of the spacelike effective potential $V_{\text{eff}}$:
\bea
\tilde{V}_{\text{eff}}(r) = -V_{\text{eff}}(r) = r^{2d-2}f(r)\,.
\eea
As in the case of spacelike surfaces, there is an accumulation surface whose location corresponds to the maximum of $\tilde{V}_{\text{eff}}$, which is the same as the minimum of $V_{\text{eff}}$. We will denote the location of this surface by $r=\hat{r}_a$. The orientation of the particle in the scattering picture depends on its turning point w.r.t to the new accumulation surface. When $r_{\min}\geq\hat{r}_a$, the volume of these timelike extremal surfaces is given by
\bea
\mathcal{V}_{\text{timelike}}= i \Omega_{d-1} \int_{r_{\min}}^{r_v} \dd r \frac{r^{2(d-1)}}{\sqrt{-f(r) r^{2(d-1)}+p^2}}\,. \label{volumeexpressiontimelike}
\eea 
The volume is imaginary, which follows from our choice of metric signature. The limits of integration reverse when $r_{\min}<\hat{r}_a$. The growth rate can then be found to be:
\bea 
\frac{\dd\mathcal{V}_{\text{timelike}}}{\dd v} = \begin{dcases}
 \frac{i\Omega_{d-1}}{2} \left(p-\sqrt{p^2-f(r_v) r_v^{2(d-1)}}\right)\, & r_{\min}\geq\hat{r}_a\\
 \frac{i\Omega_{d-1}}{2} \left(-p+\sqrt{p^2-f(r_v) r_v^{2(d-1)}}\right)\, & r_{\min}<\hat{r}_a\\
\end{dcases}.\label{volumegrowthtimelike}
\eea 
As in the spacelike case, the growth rate changes its sign when the anchor point crosses the accumulation surface. This prompts us to call the surfaces with $r_{\min} <\hat{r}_a$ \emph{timelike trapped extremal surfaces}.

More importantly, we observe that $p$ vanishes when $r_v=0$. In this case, the extremal surface simplifies and becomes the constant $r=0$ slice, and the right-hand side of the first equation of motion \eqref{eometernaltime} vanishes identically. Consequently, the growth rate of the extremal surface drops to zero. Since the singularity is the only region where the growth rate vanishes, this behavior can be used as a robust signature of the black hole singularity.

While it is perhaps not surprising that timelike extremal surfaces are sensitive to timelike singularities, their dual interpretation remains unclear. In particular, the meaning of the volume of a timelike extremal surface in the boundary theory is not immediately evident. Nonetheless, it is tempting to speculate that this may hint at a notion of ``timelike complexity'' in the dual description. Interestingly, timelike entanglement entropy has been proposed as a method for computing what is referred to as pseudo-entropy in the boundary theory \cite{Doi:2022iyj,Narayan:2022afv,Doi:2023zaf,Narayan:2023ebn}. It is natural to expect a corresponding dictionary to exist in the case of complexity. We leave this for future work.
\subsection{Kerr Black Hole}
The vanishing of the growth rate of timelike extremal surfaces at the black hole singularity is not limited to spherically symmetric black holes. We can see this by studying a $3+1$-dimensional Kerr black hole of mass $M$ and angular momentum $a M$. The metric in Boyer--Lindquist coordinates is given by
\bea
\dd s^2=-\frac{\Delta}{\rho^2}\left(\dd t-a \sin ^2 \theta \dd \phi\right)^2+\frac{\sin ^2 \theta}{\rho^2}\left[\left(r^2+a^2\right) \dd \phi-a \dd t\right]^2+\frac{\rho^2}{\Delta} \dd r^2+\rho^2 \dd \theta^2\,, \label{kerrmetriceq}
\eea 
where
\bea
\Delta=r^2-2 G M r+a^2 \quad \text { and } \quad \rho^2=r^2+a^2 \cos ^2 \theta\,.
\eea 
The black hole horizons are located at $r_{\pm}$. The ring singularity sits at $r=0$ and $\theta=\pi/2$. We can also define an Eddington-Finkelstein coordinate $v = t+r^{*}$, where
\bea
\dd r^{*} = \frac{r^2+a^2}{\Delta}\dd r\,.
\eea 
As in the RN case, we consider extremal surfaces anchored at a null surface inside the black hole. However, we now focus on axially symmetric configurations, as the black hole is rotating. We anchor the `right' endpoint of the extremal surfaces at the null surface. The other endpoint is chosen to lie at its turning point. The spacetime is not globally hyperbolic, and as a result, when the anchor point crosses the inner horizon, there are new timelike extremal surfaces. In particular, as the anchor point approaches the singularity, the timelike extremal surfaces approach the singularity. Let us see how this works out.

Consider a general axisymmetric hypersurface in the interior of the black hole. The $r$ coordinate can be parametrized by $r(t,\theta)$ and the line element of this three-manifold can be computed by substituting $\dd r = r_{,t}\dd t+r_{,\theta} \dd \theta$ into to the metric \eqref{kerrmetriceq}. Once we have the metric, we can compute the volume by integrating over the variables $t,\theta$, and $\phi$.  The full expression is given in \cite{PhysRevD.31.1267}, although we will not need it explicitly here since we are only focusing on surfaces near the ring singularity.

We therefore restrict ourselves to the case where $r$ is just a function of $\theta$. We will also assume that the surface is $t$-independent. In this case, the three-metric computation simplifies significantly, and we obtain \cite{Couch:2018phr,Chew:2020twk}
\bea
\dd \mathcal{V} =\sqrt{\rho^2\left(|\Delta|-r_{, \theta}^2\right)} \sin \theta d t \wedge d \wedge d \theta \wedge d \phi\,. \label{volumekerrelement}
\eea
The extremal hypersurfaces are given by extremizing this integral over the volume element. If $r=$ constant, the Euler-Lagrange equation is given by the simple expression
\bea
\frac{\dd}{\dd r} \left(\rho^2\Delta\right)=0\,.
\eea 
We can see that this is trivially satisfied by the ring singularity at $r=0$ and $\theta=\pi/2$. Therefore, the singularity is an extremal volume surface. Moreover, it is a \emph{timelike} extremal surface.

Now, let us compute the growth rate of the timelike surfaces as they approach the singularity.  Since $\dd v= \dd t$ on a constant $r$ slice, we can use the volume element \eqref{volumekerrelement} to obtain
\bea
\frac{\dd \mathcal{V}}{\dd v} = \int \sqrt{\rho^2 \Delta}\sin{\theta}  \ \dd \theta \wedge \dd \phi\,.
\eea 
At the singularity, the integrand vanishes, and the growth of timelike extremal surfaces goes to zero. This suggests that the growth rate of timelike or spacelike extremal surfaces can serve as a useful probe of black hole singularities in various cases.
\section{Incompleteness in Rindler Wedge}
\label{rindlersec}
Incompleteness of null geodesics does not necessarily imply the presence of a singularity where the geodesics terminate \cite{Hawking:1973uf}. In this section, we explore the alternative possibility by studying complexity growth in AdS-Rindler wedges. In global coordinates, AdS$_3$ has the metric
\bea
\dd s^2 =\frac{1}{\cos ^2 \rho}\left(-\dd \tau^2+ \dd \rho^2+\sin ^2 \rho \dd \theta^2\right)\,.
\eea
We have set the AdS length scale to one for convenience. Using the coordinate transformations, 
\bea
\begin{aligned}
 t & =\frac{ \sin \tau}{\cos \tau+\sin \rho \cos \theta}\,, \\
 x & =\frac{ \sin \theta \sin \rho}{\cos \tau+\sin \rho \cos \theta}\,, \\
 z & =\frac{ \cos \rho}{\cos \tau+\sin \rho\cos \theta}\,,
  \end{aligned}
\eea
we can rewrite the metric in the Poincar\'{e} coordinates as follows:
\bea
\dd s^2 = \frac{\dd z^2 -\dd x^{+} \dd x^{-}}{z^2}\,, \qquad \text{where} \ x^{\pm} = t\pm x\,.
\eea

\begin{figure}
  \centering
  \hspace*{-1.5cm} \includegraphics[width=0.7\linewidth]{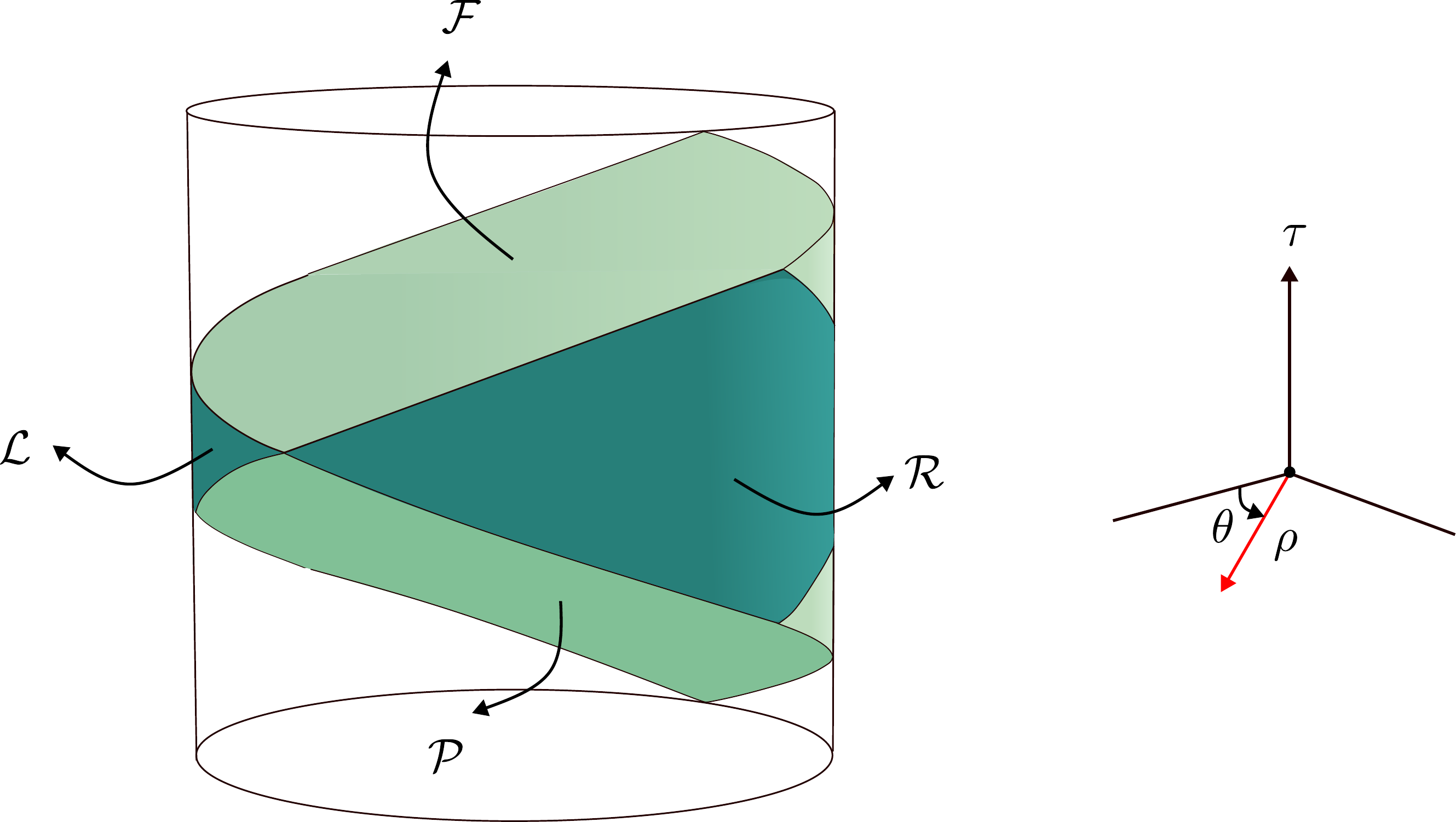}
  \caption{\small The Poincar\'{e} patch of AdS$_3$. The global coordinates cover the entire cylinder. The Poincar\'{e} coordinates cover only the solid region shown in the figure, which can be further divided into four Rindler wedges.}
  \label{rindlerwedgefig0}
  \end{figure}

  \noindent The region $z>0$ and $x^{\pm} \in (-\infty,\infty)$ defines the Poincar\'{e} patch of the spacetime. We can divide the patch into four Rindler regions depending on the signs of $(x^+,x^-)$. We label these regions $\mathcal{R},\mathcal{L},\mathcal{F}$, and $\mathcal{P}$ when the signs are (+, -), (-, +), (+, +), and (-, -) respectively (see Figure \ref{rindlerwedgefig0}). The relevance of the boundary of these wedges in the context of AdS/CFT correspondence can be found in \cite{Leutheusser:2021frk}.

It turns out that we can also cover these Rindler regions using a set of BTZ coordinates $(t,r,\phi)$ given by
\bea
e^{2 t}=\frac{x^{+}}{x^{-}}, \quad e^{2 \phi}=z^2-x^{+} x^{-}, \quad r^2=\frac{z^2-x^{+} x^{-}}{z^2}\,.\label{BTZmetriccoordeq}
\eea 
The metric in these coordinates takes the form
\bea
\dd s^2 = -\left(r^2-1\right)\dd t^2 -\frac{\dd r^2}{\left(r^2-1\right)}+ r^2\dd \phi \,. \label{adsrindlermetric}
\eea
The crucial difference between the BTZ metric and the AdS-Rindler metric in BTZ coordinates \eqref{adsrindlermetric} is that $\phi$ is non-compact in the latter case.

The AdS-Rindler horizon is at $r=1$. The Rindler regions $\mathcal{R},\mathcal{L}$ are mapped to the left and right exterior wedges of the BTZ “black hole.” In contrast, the Rindler regions $\mathcal{F}, \mathcal{P}$ are mapped to the future and past interiors, respectively. Now, let us look at the growth of extremal surfaces in this coordinate system. We will study surfaces that are symmetric in the $\phi$ coordinate. Since this direction is non-compact, we will regularize the volume by choosing appropriate constant-$\phi$ cutoff surfaces. The metric \eqref{adsrindlermetric} has the same form as \eqref{staticbhmetric}. This allows us to port the uncharged static black hole computations. A trapped extremal surface exists in the interior of the ``black hole,'' our singularity theorem tells us that null geodesics are incomplete.

The theorem is signaling the existence of the $r=0$ surface. Clearly, this surface is not a singularity in the full AdS spacetime. Nevertheless, it appears problematic from the perspective of the BTZ coordinate system. To understand the significance of the $r=0$ surface, we must look at the relationship between the Poincar\'{e} coordinates and the BTZ coordinates. 

\begin{figure}
  \centering
  \hspace*{-1.5cm} \includegraphics[width=1\linewidth]{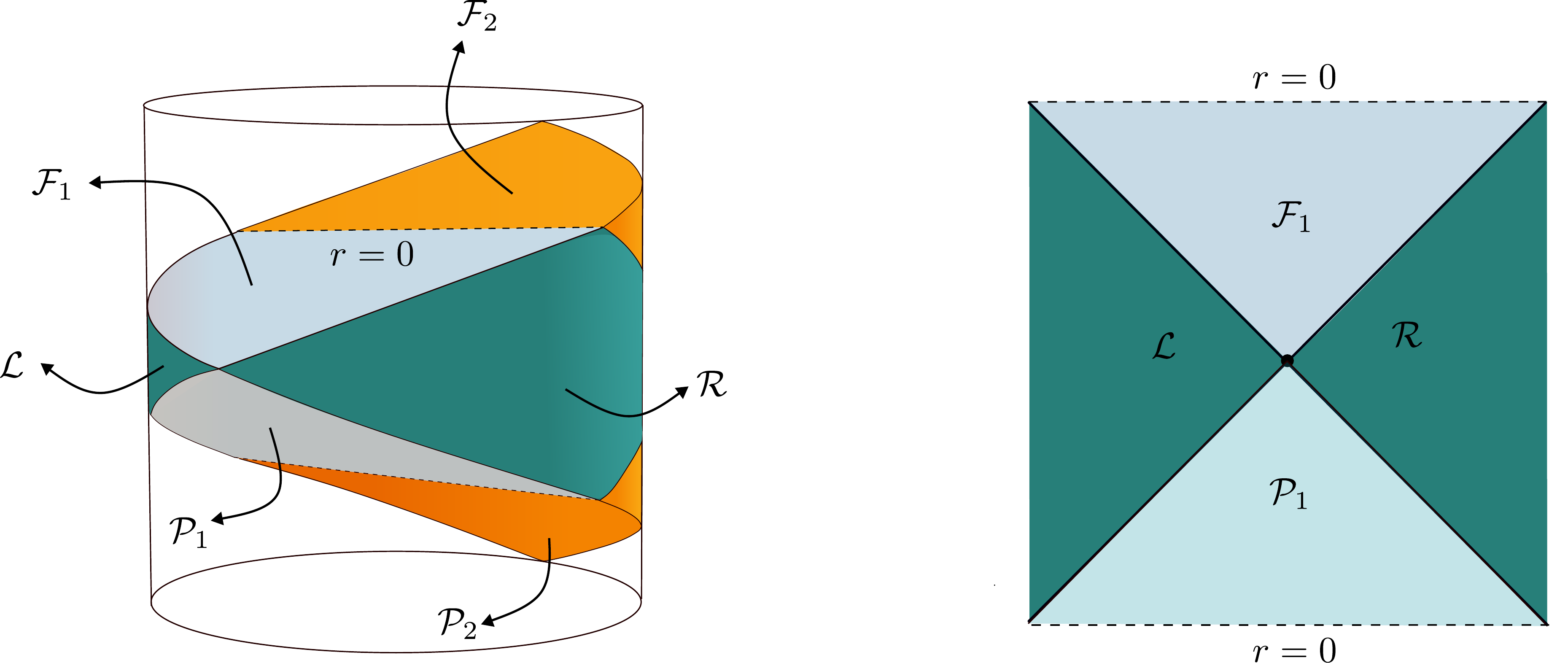}
  \caption{\small \emph{(Left)} The patch $\mathcal{F}$ is divided into two regions, $\mathcal{F}_1$ and $\mathcal{F}_2$. The BTZ coordinates cover only $\mathcal{F}_1$. \emph{(Right)} The mapping of various Rindler regions to the corresponding BTZ ``black hole.''}
  \label{rindlerwedgefig}
  \end{figure}

It turns out that the BTZ coordinates do not cover the entire Rindler region $\mathcal{F}$ \cite{Leutheusser:2021frk}. The BTZ coordinates cover only a part of $\mathcal{F}$, which we will denote by $\mathcal{F}_1$ (see Figure \ref{rindlerwedgefig}). 
 
In the region $\mathcal{F}_2$, we must use a different set of coordinates (the coordinate transformations can be found in \cite{Leutheusser:2021frk}). The surface $r = 0$ corresponds to the boundary between $\mathcal{F}_1$ and $\mathcal{F}_2$. Therefore, the geodesic incompleteness stemming from SLC indicates the inability of the coordinate patch to cover the entirety of the Poincar\'{e} patch.

We will end this discussion with an important note. The inability to conclusively identify the terminus of null geodesics as a singularity is a general feature of all singularity theorems, not merely a limitation of our own.
\section{Discussion}
\label{Discussionsection}
Our proof of geodesic incompleteness is structurally similar to that of Wall \cite{Wall:2010jtc}. However, it is important to highlight the differences. The generalized entropy defined in \eqref{genentropyeq} consists of a classical area term and a quantum entanglement entropy term. The inclusion of both terms is essential for obtaining a second law \cite{Bekenstein:1972tm}.

In contrast, holographic complexity involves the computation of a classical volume. Quantum corrections become important only when the complexity is of exponential order in the Bekenstein-Hawking entropy. This is explicitly evident in 2d models of quantum gravity \cite{Yang:2018gdb,Iliesiu:2021ari}. The absence of a quantum term significantly simplifies the computations. It is well known in the literature that the entanglement entropy term complicates the comparison of arbitrary shape deformations of generalized entropy (see \cite{Bousso:2025xyc} for more details). Our volume calculations sidestep this technical challenge by remaining purely classical.

The growth rate of extremal surfaces has been very useful in signaling the existence of a singularity inside the black hole. A natural question is: what does this object mean in a quantum theory, for instance, in the language of operator growth \cite{Parker:2018yvk}? Since we are computing the volume of a \emph{subregion} inside the black hole, a natural guess would be to compute the complexity of a subsystem. This notion of subsystem complexity should be distinguished from existing ideas in the literature, such as the one in \cite{Alishahiha:2015rta}, where the volume enclosed by the Ryu-Takayanagi (RT) surface is used to define a notion of subregion complexity in the boundary theory. In contrast, the bulk subregion considered here lies within the black hole interior. It is therefore natural to expect that the associated subregion complexity computes the complexity of boundary operators that are dual to interior degrees of freedom. These operators are inherently non-local in the boundary theory and can, for example, be constructed using the Papadodimas-Raju proposal \cite{Papadodimas:2012aq,Papadodimas:2013wnh,Papadodimas:2013jku}.

If one can make this notion of subregion complexity more precise, we can use these ideas to diagnose the existence of singularities in their holographic duals. Recent developments in the study of operator algebras suggest that the emergence of a Killing horizon in the bulk can be attributed to a particular notion of chaos in the boundary theory \cite{Gesteau:2024rpt,Ouseph:2023juq}. Complexity might provide a natural framework for understanding the emergence of black hole singularities and the interior (see, e.g., \cite{Arean:2024pzo,Caceres:2024edr}). Operator complexity, for example, has been very useful in describing aspects of the exterior geometry of the black hole \cite{Fan:2024iop,Caputa:2024sux,Dodelson:2025rng,Jeong:2024oao}. Therefore, it is natural to expect complexity to teach us something about the emergence of the black hole interior.

An important takeaway from our calculations and that of Wall \cite{Wall:2010jtc} is that the black hole singularity is deeply tied to the thermodynamic irreversibility of the second law of thermodynamics. It will be interesting to see if we can use an appropriate notion of thermodynamic irreversibility to directly motivate a boundary diagnostic of black hole singularities.

In this paper, we have ignored quantum corrections to the volume calculation. It is known that non-perturbative effects play an important role in the computation of complexity \citeleft\citen{Iliesiu:2021ari}\citepunct\citen{Gautason:2025ryg}\citeright. These corrections are expected to appear at the semiclassical level as bulk `wormholes' and are believed to account for the finiteness of the black hole Hilbert space \cite{Balasubramanian:2022gmo}. Including such effects would modify the growth rate calculations and may offer a path toward resolving the singularity.

\section*{Acknowledgements} 
I would like to thank Friðrik Freyr Gautason, Watse Sybesma, and Lárus Thorlacius for useful discussions. This work was supported in part by the Icelandic Research Fund under grant 228952-053 and by a doctoral grant from The University of Iceland Science Park. 

\bibliographystyle{JHEP}
\bibliography{refs}

\end{document}